\documentclass[10pt,journal,twocolumn,twoside,letter]{IEEEtran}

\pdfpageattr {/Group << /S /Transparency /I true /CS /DeviceRGB>>} 

\usepackage{amsmath}
\usepackage{amssymb}
\usepackage{amsthm}
\usepackage{latexsym}
\usepackage{verbatim}
\usepackage{graphicx}

\newtheorem{thm}{Theorem}
\newtheorem{prop}{Proposition}

{\begin{list}               
    {$\bullet$ \hfill}{
        \setlength{\leftmargin}{\parindent}
        \setlength{\parsep}{0.04\baselineskip}
        \setlength{\itemsep}{0.5\parsep}
        \setlength{\labelwidth}{\leftmargin}
        \setlength{\labelsep}{0em}}
    }
{\end{list}}

\providecommand{\cref}[1]{Chapter~\ref{chap:#1}}

\providecommand{\C}{\ensuremath{\mathbb{C}}}

\providecommand{\Z}{\ensuremath{\mathbb{Z}}}

\providecommand{\abs}[1]{\left|#1\right|}
\providecommand{\norm}[1]{\lVert#1\rVert}
\providecommand{\inprod}[1]{\left\langle#1\right\rangle}

\providecommand{\bydef}{\overset{\text{def}}{=}}
\providecommand{\diag}{\mathop{\mathrm{diag}}}

\providecommand{\e}{\ensuremath{\mathrm{e}}}
\providecommand{\I}{\ensuremath{j}}
\providecommand{\epi}[1]{\e^{\hspace{1pt}\I{#1}}}
\providecommand{\eni}[1]{\e^{-\I{#1}}}

\renewcommand{\vec}[1]{\ensuremath{\mathbf{#1}}}
\providecommand{\mat}[1]{\ensuremath{\mathbf{#1}}}

\providecommand{\wt}[1]{\ensuremath{\widetilde{#1}}}

\providecommand{\mA}{\mat{A}} 
\providecommand{\mC}{\mat{C}} 
 
\providecommand{\mI}{\mat{I}}  
\providecommand{\mK}{\mat{K}}

 \providecommand{\mU}{\mat{U}} 
 \providecommand{\mT}{\mat{T}}

\providecommand{\va}{\vec{a}}

 \providecommand{\vn}{\vec{n}} 
 \providecommand{\vp}{\vec{p}}
\providecommand{\vq}{\vec{q}} \providecommand{\vr}{\vec{r}}
\providecommand{\vs}{\vec{s}}
\providecommand{\vt}{\vec{t}}
 \providecommand{\vw}{\vec{w}}

\providecommand{\vx}{\vec{x}} \providecommand{\vy}{\vec{y}}

\usepackage{bm}
\usepackage[noadjust]{cite}
\usepackage{booktabs} 
\usepackage{mathrsfs}
\usepackage[pagebackref=false,
            citecolor=blue,
            urlcolor=blue,
            colorlinks=true]{hyperref}
\usepackage{units}
\usepackage{placeins}       
\usepackage{balance} 
\usepackage{nicefrac}
\usepackage{xcolor}

\usepackage{wasysym}

\usepackage{float}

\renewcommand{\H}{H}
\newcommand{\T}{T}

\newcommand{\SNR}{\mathsf{SNR}}
\newcommand{\SINR}{\mathsf{SINR}}
\newcommand{\UDR}{\mathsf{UDR}}
\newcommand{\E}{\mathbb{E}}
\newcommand{\ejw}{\epi{\omega}}
\newcommand{\sinc}{\operatorname{sinc}}

\definecolor{nice_blue}{rgb}{0.5 0.1 0.9}
\definecolor{note_color}{rgb}{1 0.3 0.3}
\newcommand{\rev}[1]{{#1}}

\DeclareMathOperator*{\minimize}{minimize}
\DeclareMathOperator*{\maximize}{maximize}

\begin{document}
\title{Raking the Cocktail Party}

\author{Ivan~Dokmani\'c,~\IEEEmembership{Student Member,~IEEE,}
        Robin~Scheibler,~\IEEEmembership{Student Member,~IEEE,}
        and~Martin~Vetterli,~\IEEEmembership{Fellow,~IEEE}
\thanks{Authors are with LCAV-EPFL. This work was supported by the ERC Advanced
  Investigators Grant: Sparse Sampling: Theory, Algorithms and Applications
  SPARSAM no. 247006, and a Google Doctoral Fellowship.}}

\markboth{Submitted to IEEE Journal on Selected Topics in Signal Processing,~July~2014}%
{Dokmani\'c \MakeLowercase{\textit{et al.}}: Acoustic Rake Receivers}

\maketitle

\begin{abstract}
  We present the concept of an acoustic rake receiver---a microphone
  beamformer that uses echoes to improve the noise and interference
  suppression. The rake idea is well-known in wireless communications; it
  involves constructively combining different multipath components that arrive
  at the receiver antennas. Unlike spread-spectrum signals used in wireless
  communications, speech signals are not orthogonal to their shifts.
  \rev{Therefore, we focus on the spatial structure, rather than temporal.}
  Instead of explicitly estimating the channel, we create correspondences
  between early echoes in time and image sources in space. These multiple
  sources of the desired and the interfering signal offer additional spatial
  diversity that we can exploit in the beamformer design.

  We present several ``intuitive'' and optimal formulations of acoustic rake
  receivers, and show theoretically and numerically that the rake formulation
  of the maximum signal-to-interference-and-noise beamformer offers
  significant performance boosts in terms of noise and interference
  suppression. Beyond signal-to-noise ratio, we observe gains in terms of the
  \emph{perceptual evaluation of speech quality} (PESQ) metric for the speech
  quality. We accompany the paper by the complete simulation and processing
  chain written in Python. The code and the sound samples are available online
  at
  \url{http://lcav.github.io/AcousticRakeReceiver/}.
\end{abstract}



\section{Introduction}

\IEEEPARstart{R}{ake} receivers take advantage of multipath propagation,
instead of trying to mitigate it. The basic idea of the rake receivers
(habitually used in wireless communications) is to coherently add the
multipath components, and thus increase the effective signal-to-%
noise ratio ($\SNR$). The original scheme was developed for
single-input-single-output systems \cite{Price:1958ec}, and it was later
extended to arrays of antennas \cite{Naguib:1997ic,Khalaj:1994fc} that exploit
spatial diversity. By using antenna arrays and spatial processing, multipath
components that would otherwise not be resolvable because they arrive at
similar times, become resolvable because they arrive from different
directions.
  
In spite of the success of the rake receivers in wireless communications, the
principle has not received significant attention in room acoustics.
Nevertheless, constructive use of echoes in rooms to improve beamforming has
been mentioned in the literature
\cite{Annibale:2011wg,ODonovan:2010gw,Jan:1995gs}. In particular, the term
\emph{acoustic rake receiver} (ARR) was used in the SCENIC project proposal
\cite{Annibale:2011wg}. 

The list of ingredients for ARRs in room acoustics is similar as in wireless
communications: a wave (acoustic instead of electromagnetic) propagates in
space; reflections and scattering cause the wave to arrive at the receiver
through multiple paths in addition to the direct path, and these multipath
components all contain the source waveform.



\rev{The main difference is that in room acoustics we do not get to design the
input signal. Spreading sequences used in CDMA are designed to be orthogonal
to their shifts, which facilitates the multipath channel estimation; this
orthogonality is not exhibited by speech. Moreover, speech segments are very
long with respect to the time between the two consecutive echoes.}


On the contrary, there are no significant differences in terms of the spatial
structure. If we know where the echoes are coming from, we can design spatial
processing algorithms---for example beamformers---that use multiple copies of
the same signal arriving from different directions.

Imagine first that we know the room geometry. Then, if we localize the source,
we can predict where its echoes will come from using simple geometric rules
\cite{Allen:1979ua,Borish:1984uu}. Localizing the direct signal in a
reverberant environment is a well-understood problem
\cite{Ward:2003cf}. What is more, we do not need to know the room shape in
detail---locations of the most important reflectors (ceiling, floor, walls)
suffice to localize the major echoes. In many cases this knowledge is readily
available from the floor plans or measurements. In ad-hoc deployments, the
room geometry may be difficult to obtain. If that is the case, we can first
perform a calibration step to learn it. An appealing method to infer the room
geometry is by using sound, as was demonstrated recently
\cite{Ribeiro:2012ge,Antonacci:2012hc,Dokmanic:2011vc,Dokmanic:2013dz}.

We may still be able to take advantage of the echoes without estimating the
room geometry. Note that we are not after the room geometry itself; rather, we
only need to know where the early echoes are coming from.
\rev{Echoes can be seen as signals emitted by \emph{image sources}---mirror
images of the true source across reflecting walls \cite{Allen:1979ua}. Knowing
where the echoes are coming from is equivalent to knowing where the image
sources are.}

Image source localization can be solved, for example, by
\emph{echo sorting} as described in
\cite{Dokmanic:2013dz}. Alternatively, O'Donovan, Duraiswami and Zotkin
\cite{ODonovan:2010gw} propose to use an \emph{audio camera} with a large
number of microphones to find the images. Once the image sources are localized
(in a calibration phase or otherwise), we can predict their movement using
geometrical rules, as discussed in Section
\ref{sec:finding_and_tracking_the_echoes}. \rev{Thus, the acoustic raking is a
multi-stage process comprising image source localization, tracking, and
beamforming weight computation. The complete block diagram is shown in
Fig.~\ref{fig:arr_block_diagram}.}


It is interesting to note the analogy between the ARRs and the human auditory
perception. It is well established that the early echoes improve speech
intelligibility \cite{Bradley:2003jf,Lochner:1964do}. In fact, adding energy
in the form of early echoes (approximately within the first 50 ms of the room
impulse response (RIR)) is equivalent to adding the same energy to the direct
sound
\cite{Bradley:2003jf}. This observation suggests new designs for indoor
beamformers, with different choices of performance measures and reference
signals. A related discussion of this topic is given by Habets and co-authors
\cite{Habets:2010iw}, who examine the tradeoff between dereverberation and
denoising in beamforming. \rev{In addition to the standard $\SNR$, we propose
to use the useful-to-detrimental ratio ($\UDR$), first defined by Lochner and
Burger \cite{Lochner:1964do}, and used by Bradley, Sato and Picard
\cite{Bradley:2003jf}. We generalize $\UDR$ to a scenario with interferers,
defining it as the the ratio of the direct and early reflection energy to
the energy of the noise and interference.}

ARRs focus on the early part of the RIR, trying to concentrate the energy
contained in the early echoes. In that regard, there are similarities between
ARRs and channel shortening \cite{Thomas:2011iv,Zhang:2010vq}. Channel
shortening produces filters that are much better behaved than complete
inversion, \emph{e.g.}, by the multiple-input-output-theorem (MINT)
\cite{Miyoshi:1988ig, Furuya:2001ta}. Nevertheless, it is assumed that we know
the acoustic impulse responses between the sources and the microphones. In
contrast to channel shortening, as well as other methods assuming this
knowledge \cite{Benesty:2007hj,Miyoshi:1988ig}, we never attempt the difficult
task of estimating the impulse responses. Our task is simpler: we only need to
detect the early echoes, and lift them to 3D space as image sources.

\subsection{Main Contributions and Limitations}

\begin{figure}
\centering
\includegraphics[width=3.5in]{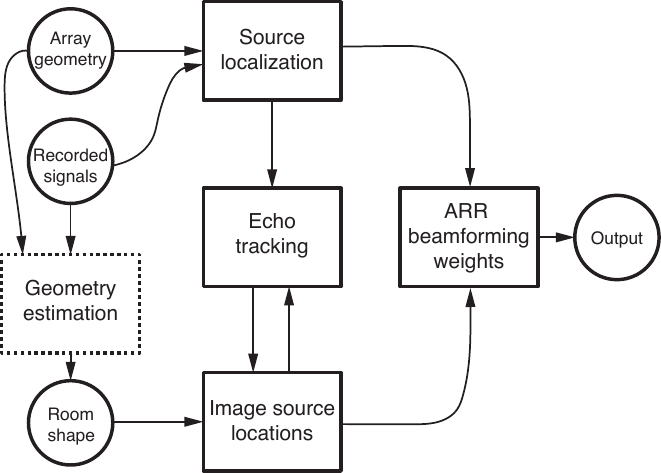}
\caption{A block diagram for acoustic rake receivers. In this paper, we focus
on ARR beamforming weight computation, and we briefly discuss echo tracking
and image source localization. The geometry estimation block is optional (room
geometry could be known in advance), hence the dashed box.}
\label{fig:arr_block_diagram}
\end{figure}

We introduce the acoustic rake receiver (ARR) as the echo-aware microphone
beamformer. We present several formulations with different properties, and
analyze their behavior theoretically and numerically. The analysis shows
that ARRs lead to significantly improved $\SNR$ and interference cancellation
when compared with standard beamformers that only extract the direct path.
ARRs can suppress interference in cases when conventional beamforming is bound
to fail, for example when an interferer is occluding the desired source (for a
sneak-peak, fast forward to Fig.
\ref{fig:beam_scenarios}). We present optimal formulations that outperform the
earlier delay-and-sum (DS) approaches \cite{Jan:1995gs}, especially when
interferers are present. \rev{Significant gains are observed not only in terms
of signal-to-interference-and-noise ratio ($\SINR$) and $\UDR$, but also in
terms of perceptual evaluation of speech quality (PESQ)
\cite{Rix:2001bv}.}

\rev{The raking microphone beamformers are particularly well-suited to
extracting the desired speech signal in the presence of interfering sounds, in
part because they can focus on echoes of the desired sound and cancel the
echoes of the interference. The analogous human capacity to focus on a
particular acoustic stimulus while not perceiving other, unwanted sounds is
called the \emph{cocktail party effect} \cite{Haykin:2005cv}. The title of
this paper was inspired by that analogy.}

We design and apply the ARRs in the frequency domain. Frequency domain
formulation is simple and concise; it allows us to focus on objective gains
from acoustic raking. Time-domain designs \cite{Herbordt:2003kea} offer better
control over the impulse responses of the beamforming filters, but they are
out of the scope of this paper. For a recent time-domain approach to ARR, see
\cite{Scheibler:2015xx}.

Let us also mention some limitations of our results. For clarity, the
numerical experiments are presented in a 2D ``room'', and as such are directly
applicable to planar (\emph{e.g.} linear or circular) arrays. Extension to 3D
arrays is straightforward. We do not discuss robust formulations that address
uncertainties in the array calibration. Microphones are assumed to be ideally
omni-directional with a flat frequency response. Except for Section
\ref{sec:finding_and_tracking_the_echoes}, we assume that the locations of the
image sources are known. We explain how to find the image sources when the
room geometry is either known or unknown, but we do not provide a deep
overview of the geometry estimation techniques. To this end, we suggest a
number of references for the interested reader. We consider the walls to be
flat-fading; in reality, they are frequency selective. We do not discuss the
estimation of various covariance matrices \cite{Carlson:1988hn}.

The results in this paper are reproducible. Python (NumPy)
\cite{Oliphant:2007dm} code for the beamforming routines, for the STFT
processing engine, and to generate the figures and the sound samples is
available online at \url{http://lcav.github.io/AcousticRakeReceiver/}.

\subsection{Paper Outline}

In Section \ref{sec:notation} we explain the notation and the signal model
used in the paper. A brief overview of the relevant beamforming techniques and
performance analysis is given in Section \ref{sec:beamforming}. We formulate
the acoustic rake receiver in Section \ref{sec:acoustic_rake}, and we present
a theoretical and numerical analysis of the corresponding beamformers. Section
\ref{sec:finding_and_tracking_the_echoes} explains how to locate the image
sources, and comments on localizing the direct source. Numerical experiments
are presented in Section
\ref{sec:numerical_experiments}.


\section{Notation and Signal model}

We denote all matrices by bold uppercase letters, for example $\mA$, and all
vectors by bold lowercase letters, for example $\vx$. The Hermitian transpose
of a matrix or a vector is denoted by $(\, \cdot \,)^\H$, as in $\mA^\H$, and the
Euclidean norm of a vector by $\norm{\, \cdot \,}$, that is, $\norm{\vx} \bydef
(\vx^\H \vx)^{\nicefrac{1}{2}}$.

\label{sec:notation}

Suppose that the desired source of sound is at the location $\vs_0$ in a room.
Sound from this source arrives at the microphones located at $[\vr_m]_{m=1}^M$
via the direct path, but also via the echoes from the walls. The echoes can be
replaced by the image sources---mirror images of the true sources across the
corresponding walls---according to the image source model \cite{Allen:1979ua,
Borish:1984uu}. An important consequence is that instead of modeling the
source of the desired or the interefering signal as a single point in a room,
we can model it as a collection of points in free space. A more detailed
discussion of the image source model is given in Section
\ref{sec:finding_and_tracking_the_echoes}.

Denote the signal emitted by the source $\wt{x}[n]$ (\emph{e.g.} the speech
signal). Then all the image sources emit $\wt{x}[n]$ as well, and the signals
from the image sources reach the microphones with the appropriate delays. In
our application, the essential fact is that the echoes correspond to image
sources. We denote the image source positions by $\vs_k$, $1 \leq k \leq K$,
where $K$ denotes the largest number of image sources considered. Note that we
do not care about the sequence of walls that generates $\vs_k$, nor do we care
about how many walls are in the sequence. For us, all $\vs_k$ are simply
additional sources of the desired signal. The described setup is illustrated
Fig.~\ref{fig:notation}.

Suppose further that there is an interferer at the location $\vq_0$ (for
simplicity, we consider only a single interferer). The interferer emits the
signal $\wt{z}[n]$, and its image sources emit $\wt{z}[n]$ as well. Similarly
as for the desired source, we denote by $\vq_k$, $1 \leq k \leq K'$ the
positions of the interfering image sources, with $K'$ being the largest number
of interfering image source considered. The model mismatch (\emph{e.g.}, the
image sources of high orders and the late reverberation) and the noise are
absorbed in the term $\wt{n}_m[n]$.

\begin{table}[t!]
  \centering

  \caption{Summary of notation.}

  \begin{tabular}{@{}ll@{}}
  \toprule
  {\bf Symbol} & {\bf Meaning} \\
  \midrule
  $M$ &
  Number of microphones \\
  $\vr_m$ &
  Location of the $m$th microphone \\
  $\vs_0$ &
  Location of the desired source \\
  $\vs_i$ &
  Location of the $i$th image of the desired source ($i \geq 1$) \\
  $\vq_0$ &
  Location of the interfering source \\
  $\vq_i$ &
  Location of the $i$th image of the interfering source ($i \geq 1$) \\
  $x(\ejw)$ &
  Spectrum of the sound from the desired source \\
  $z(\ejw)$ &
  Spectrum of the sound from the interfering source \\
  $\vw(\ejw)$ &
  Vector of beamformer weights \\
  $K$ &
  Number of considered desired image sources \\
  $K'$ &
  Number of considered interfering image sources \\
  $a_m(\vs)$ &
  $m$th component of the steering vector for a source at $\vs$ \\
  $y_m$ &
  Signal picked up by the $m$th microphone \\
  $\|\,\cdot\,\|$ &
  Euclidean norm, $\|\vec{x}\| = (\sum |x_i|^2)^{\nicefrac{1}{2}}$. \\
  \bottomrule
  \end{tabular}
  \label{tab:notation-summary}
\end{table}

\begin{figure}
  \centering
  \includegraphics{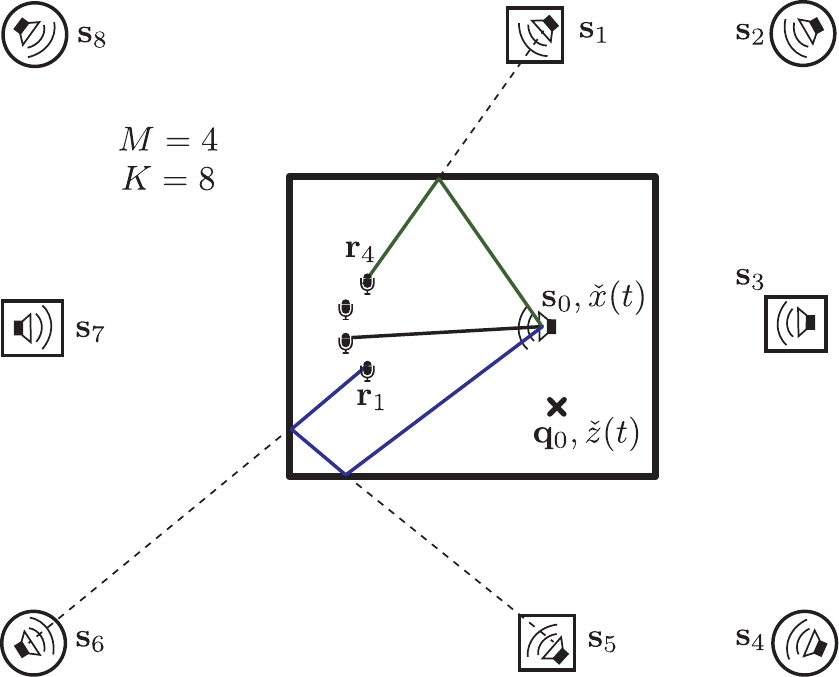}
  \caption{Illustration of the notation and concepts. Echoes of the desired
  signal emitted at $\vs_0$ can be modeled as a direct sound coming from the
  image sources of $\vs_0$. Two generations of image sources are illustrated:
  first ($\vs_1, \vs_3, \vs_5, \vs_7$) and second ($\vs_2, \vs_4, \vs_6,
  \vs_8$), as well as the corresponding \emph{sound rays} for $\vs_5$ and
  $\vs_6$. The interferer is located at $\vq_0$ (its image sources are not
  shown), and the microphones are located at $\vr_1, \ldots, \vr_4$.}
  \label{fig:notation}
\end{figure}

The signal received by the $m$th microphone is then a sum of convolutions
\begin{equation}
\begin{aligned}
  \wt{y}_m[n] &= \sum_{k=0}^K \big(\wt{a}_m(\vs_k) \ast \wt{x}\big)[n] \\
                 &+ \sum_{k=0}^{K'} \big(\wt{a}_m(\vq_k) \ast \wt{z}\big)[n]
                  + \wt{n}_m[n],
\end{aligned}
\end{equation}
where $\wt{a}_m(\vs_k)$ denotes the impulse response of the channel between
the source located at $\vs_k$ and the $m$th microphone---in this case a delay
and a scaling factor.

We design and analyze all the beamformers in the frequency domain. That
is, we will be working with the DTFT of the discrete-time signal $\wt{x}$,
\begin{equation}
  x(\ejw) \bydef \sum_{n \in \Z} \wt{x}[n] \, \eni{\omega n}.
\end{equation}
In practical implementations, we use the discrete-time short-time Fourier
transform (STFT). More implementation details are given in Section
\ref{sec:numerical_experiments}.

Using these notations, we can write the signal picked up by the $m$th
microphone as
\begin{equation}
\begin{aligned}
  y_m(\ejw) &= \sum_{k=0}^K a_m(\vs_k, \Omega) x(\ejw) \\
                &+ \sum_{k=0}^{K'} a_m(\vq_k, \Omega) z(\ejw)
                + n_m(\ejw),
\end{aligned}
\end{equation}
where $n_m(\ejw)$ models the noise and other errors, and $a_m(\vs_k,
\Omega)$ denotes the $m$th component of the steering vector for the source
$\vs_k$. The steering vector is the Fourier transform of the continuous
version of the impulse response $\wt{a}(\vs_k)$, evaluated at the frequency
$\Omega$. The discrete-time frequency $\omega$ and the continuous-time frequency
$\Omega$ are related as $\omega = \Omega \, T_s$, where $T_s$ is the sampling
period. The steering vector is then simply $\va(\vs_k, \Omega) = [a_m(\vs_k,
\Omega)]_{m=0}^{M-1}$.

We can write out the entries of the steering vectors explicitly for a point
source in free space. They are given as the appropriately scaled free-space
Green's functions for the Helmholtz equation \cite{Duffy:2001uy},

\begin{equation}
  \label{eq:steering}
  a_m(\vs_k, \Omega) = \frac{\alpha_k}{4\pi \norm{\vr_m - \vs_k}} \eni{\kappa \norm{\vr_m - \vs_k}},
\end{equation}
where we define the wavenumber as $\kappa \bydef \Omega/c$, and $\alpha_k$ is
the attenuation corresponding to $\vs_k$.


Using vector notation, the microphone signals can be written concisely as
\begin{equation}
\label{eq:signal_model_freq}
\begin{aligned}
  \vy(\ejw) = \mA_s(\ejw) \vec{1} x(\ejw) + \mA_q(\ejw) \vec{1} z(\ejw) + \vn(\ejw),
\end{aligned}
\end{equation}
where $\mA_s(\ejw) \bydef [\va(\vs_1, \Omega), \ldots, \va(\vs_K, \Omega)]$,
$\mA_q(\ejw) \bydef [\va(\vq_1, \Omega), \ldots, \va(\vq_{K'}, \Omega)]$, and
$\vec{1}$ is the all-ones vector. From here onward, we suppress the frequency
dependency of the steering vectors and the beamforming weights to reduce the
notational clutter.
 
\section{Beamforming Preliminaries} 
\label{sec:beamforming}

Microphone beamformers combine the outputs of multiple microphones in order to
achieve spatial selectivity, thereby suppressing noise and interference
\cite{VanVeen:1988hg}. In the frequency domain, a beamformer forms a linear
combination of the microphone outputs to yield the output $u$. That is,

\begin{equation}
\label{eq:beamforming}
\begin{aligned}
  u
  = \vw^\H \vy 
  = \vw^\H \mA_s \vec{1} x
   + \vw^\H \mA_q \vec{1} z + \vw^\H \vn,
\end{aligned}
\end{equation}
where the vector $\vw \in \C^{M}$ contains the beamforming weights.

The weights $\vw$ are often selected so that they optimize some design
criterion. Common examples of beamformers are the delay-and-sum (DS)
beamformer, minimum-variance-distortionless-response (MVDR) beamformer,
maximum-signal-to-interference-and-noise (Max-SINR) beamformer, and
minimum-mean-squared-error (MMSE) beamformer. In this paper we discuss the
rake formulation of the DS and the Max-SINR beamformers; for completeness, we
first describe the non-raking variants.

\subsection{Delay-and-Sum Beamformer}

DS is the simplest and often quite effective beamformer \cite{VanVeen:1988hg}.
Assume that we want to listen to a source at $\vs$. Then we form the the DS
beamformer by compensating the propagation delays from the source $\vs$ to the
microphones $\vr_m$,
\begin{align}
  \label{eq:DS_definition}
  u_{\text{DS}} &=
  \frac{1}{M} \sum_{m=0}^{M-1} y_m \epi{\kappa \norm{\vr_m - \vs}} \\
  \label{eq:DS_definition_second}
  &\approx \frac{x}{4\pi\|\bar{\vr}-\vs\|} + \frac{1}{M} \sum_{m=0}^{M-1} n_m,
\end{align}
where $\bar{\vr} = \frac{1}{M} \sum_{m=0}^{M-1} \vr_m$ denotes the center of
the array. The beamforming weights can be read out from
\eqref{eq:DS_definition} as
\begin{equation}
  \label{eq:delay_and_sum_weights}
  \vw_{\text{DS}} = {\va(\vs)} \, \big/ \, {\norm{\va(\vs)}},
\end{equation}
where we used the definition of $y_m$ \eqref{eq:signal_model_freq} and the
definition of the steering vector \eqref{eq:steering}. We can see from \eqref{eq:DS_definition_second} that if $\vn \sim
\mathcal{N}(\vec{0}, \sigma^2 \mI_M)$, then the output noise is distributed
according to $\mathcal{N}(0, \sigma^2 / M)$, that is, we obtain an $M$-fold
decrease in the noise variance at the output with respect to any reference
microphone.




\subsection{Maximum Signal-to-Interference-and-Noise Ratio Beamformer}

The $\SINR$ is an important figure of merit used to assess the performance of
ARRs, and to compare it with the standard non-raking beamformers. It is
computed as the ratio of the power of the desired output signal to the power
of the undesired output signal.
The desired output signal is the output signal due to the desired source,
while the undesired signal is the output signal due to the interferers and noise.

For a desired source at $\vs$, and an interfering source at $\vq$, we can write
\begin{equation}
  \label{eq:sinr}
  \SINR \bydef \frac{\E \abs{\vw^\H \va(\vs) x}^2}{\E \abs{\vw^\H (\va(\vq) z + \vn)}^2} 
   = \sigma_x^2 \frac{\vw^\H \va(\vs) \va(\vs)^\H \vw}{\vw^\H \mK_{nq} \vw},
\end{equation}
where $\mK_{nq}$ is the covariance matrix of the noise and the interference.

It is compelling to pick $\vw$ that maximizes the $\SINR$ \eqref{eq:sinr} \cite{VanVeen:1988hg}.
The maximization can be solved by noting that the rescaling of the beamformer
weights leaves the $\SINR$ unchanged. This means that we can minimize the
denominator subject to numerator being an arbitrary constant. The solution is
given as
\begin{equation}
  \label{eq:snr_beamformer}
  \vw_\text{SINR} = \frac{\mK_{nq}^{-1} \va_s}{\va_s^\H \mK_{nq}^{-1} \va_s}.
\end{equation}

Using the definition \eqref{eq:sinr}, we can derive the $\SINR$ for the Max-SINR beamformer as
\begin{equation}
  \begin{aligned}
    \SINR 
    = \sigma_x^2 \va_s^\H \mK_{nq}^{-1} \va_s.
  \end{aligned}
\end{equation}
Because $\mK_{nq}^{-1}$ is a Hermitian symmetric positive definite matrix, it
has an eigenvalue decomposition as $\mK_{nq}^{-1} = \mU^\H \mat{\Lambda} \mU$,
where $\mU$ is unitary, and $\mat{\Lambda}$ is diagonal with positive entries.
We can write $\va^\H \mK_{nq}^{-1} \va = (\mU \va)^\H \mat{\Lambda} (\mU
\va)$. Because $\norm{\mU \va}^2 = \norm{\va}^2$, and because $\mat{\Lambda}$ is
positive, increasing $\norm{\va}^2$ typically leads to an increased $\SINR$,
although we can construct counterexamples. This will be important when we
discuss the $\SINR$ gain of the Rake-Max-SINR beamformer.


\section{Acoustic Rake Receiver}
\label{sec:acoustic_rake}

In this section, we present several formulations of the ARR. The Rake-DS
beamformer is a straightforward generalization of the conventional DS
beamformer. The one-forcing beamformer implements the idea of steering a fixed
beam power towards every image source, while trying to minimize the
interference and noise. The Rake-Max-SINR and Rake-Max-UDR beamformers
optimize the corresonding performance measures; we show in Section
\ref{sec:numerical_experiments} that the Rake-Max-SINR beamforming performs
best (except, as expected, in terms of $\UDR$).

\subsection{Delay-and-Sum Raking} 
\label{sub:delay_and_sum_raking}

If we had access to every echo separately (\emph{i.e.} not summed with all the
other echoes), we could align them all to maximize the performance gain.
Unfortunately, this is not the case: each microphone picks up the convolution
of speech with the impulse response, which is effectively a sum of overlapping
echoes of the speech signal. If we only wanted to extract the direct path, we
would use the standard DS beamformer
\eqref{eq:delay_and_sum_weights}. To build the Rake-DS receiver, we create a
DS beamformer for every image source, and average the outputs,
\begin{equation}
  \label{eq:rake-ds-sum}
  \frac{1}{K+1} \sum_{k=0}^K \frac{\alpha_k'}{M} \sum_{m=0}^{M-1} y_m \epi{\kappa \norm{\vr_m - \vs_k}},
\end{equation}
where $\alpha_k' \bydef \alpha_k / (4\pi\norm{\vr_m - \vs_k})$. We read out
the beamforming weights from \eqref{eq:rake-ds-sum} as
\begin{equation}
  \vw_\text{R-DS} = \frac{1}{\norm{\sum_k \va(\vs_k)}}\sum_{k=0}^K \va(\vs_k) = \frac{\mA_s \vec{1}}{\norm{\mA_s \vec{1}}},
\end{equation}
where we chose the scaling in analogy with \eqref{eq:delay_and_sum_weights}
(scaling of the weights does not alter the output $\SINR$). It can be seen
that this is just a scaled sum of the steering vectors for each image source.


\subsection{One-Forcing Raking} 
\label{sub:one_forcing_raking}

A different approach, based on intuition, is to design a beamformer that
listens to all $K$ image sources with the same power, and at the same time
minimizes the noise and interference energy:
\begin{equation}
\begin{aligned}
  &\minimize_{\vw \in \C^M}~ \E \abs{\sum_{k=0}^{K'} \vw^\H \va(\vq_k)z + \vw^\H \vn}^2 \\
  &\text{subject to~} \vw^\H \va(\vs_k) = 1, \forall~0 \leq k \leq K.
\end{aligned}
\end{equation}

Alternatively, we may choose to null the interfering source and its image
sources. Both cases are an instance of the standard
linearly-constrained-minimum-variance (LCMV) beamformer \cite{Frost:1972bq}.
Collecting all the steering vectors in a matrix, we can write the constraint
as $\vw^\H \mA_s = \vec{1}^\T$. The solution can be found in closed form as
\begin{equation}
  \vw_\text{OF} = \mK_{nq}^{-1} \mA_s (\mA_s^\H \mK_{nq}^{-1} \mA_s)^{-1} \vec{1}_M.
\end{equation}

A few remarks are in order. First, with $M$ microphones, it does not make
sense to increase $K$ beyond $M$, as this results in more constraints than
degrees of freedom. Second, using this beamformer is a bad idea if there is an
interferer along the ray through the microphone array and any of image
sources.

As with all LCMV beamformers, adding linear constraints uses up degrees of
freedom that could otherwise be used for noise and interference suppression.
It is better to let the ``beamformer decide'' or ``the beamforming procedure
decide'' on how to maximize a well-chosen cost function; one such procedure is
described in the next subsection.


\begin{table*}[htpb]
  \centering

  \caption{Summary of beamformers.}

  \begin{tabular}{@{}llp{0.3\linewidth}@{}}
    \toprule
    {\bf Acronym} & {\bf Description} & {\bf Beamforming Weights} \\
    \midrule
    DS &
      Align delayed copies of signal at the microphone &
      $\vw_{\text{DS}} = {\va(\vs)} / {\norm{\va(\vs)}}$ \medskip \\

    Max-SINR &
      $\text{max.~} {\vw^\H \va_s \va_s^\H \vw} / (\vw^\H \mK_{nq} \vw)$ &
      $\vw_{\text{SINR}} = {\mK_{nq}^{-1} \va_s} / (\va_s^\H \mK_{nq}^{-1} \va_s)$ 
      \medskip \\

    Rake-DS &
      Weighted average of DS beamformers over image sources &
      $\vw_{\text{R-DS}} = {\mA_s \vec{1}} / \norm{\mA_s \vec{1}}$
      \medskip \\

    Rake-OF &
      $\text{min.~}  \E \abs{\sum_{k=0}^{K'} \vw^\H \va(\vq_k)z + \vw^\H \vn}^2, 
       \text{ s.t.~} \vw^\H \mA_s = \vec{1}^\T
      $ &
      $\vw_\text{OF} = \mK_{nq}^{-1} \mA_s (\mA_s^\H \mK_{nq}^{-1} \mA_s)^{-1} \vec{1}_M$
      \medskip \\

    Rake-Max-SINR &
      $\text{max.~} {\E \abs{\sum_{k=0}^K \vw^\H \va(\vs_k)x}^2} \big/ \ \E \abs{\sum_{k=0}^{K'} \vw^\H \va(\vq_k)z + \vw^\H \vn}^2$ &
      $\vw_{\text{R-SINR}} = {\mK_{nq}^{-1} \mA_s \vec{1}} / (\vec{1}^\H \mA_s^\H \mK_{nq}^{-1} \mA_s \vec{1})$
      \medskip \\

    Rake-Max-UDR &
      $\text{max.~} {\E \sum_{k=0}^K \abs{\vw^\H \va(\vs_k)x}^2} \big/ \ \E \abs{\vw^\H \sum_{k=0}^{K'} \va(\vq_k)z + \vw^\H \vn}^2$ &
      $\vw_{\text{R-UDR}} = \mC^{-1} \wt{\vw}_\text{max}((\mC^{-1})^\H \mA_s \mA_s^\H \mC^{-1})$
      \\
    \bottomrule
  \end{tabular}
  \label{tab:beamformers}
\end{table*}

\subsection{Max-SINR Raking} 
\label{sub:max_sinr_raking}

The main workhorse of the paper is the Rake-Max-SINR. We compute the weights so
as to maximize the $\SINR$, taking into account the echoes of the desired signal,
and the echoes of the interfering signal,

\begin{equation}
  \label{eq:rake-maxsinr}
  \maximize_{\vw \in \C^M}~ \frac{\E \abs{\sum_{k=0}^K \vw^\H \va(\vs_k)x}^2}{\E \abs{\sum_{k=0}^{K'} \vw^\H \va(\vq_k)z + \vw^\H \vn}^2}.
\end{equation}

The logic behind this expression can be summarized as follows: we present the
beamforming procedure with a set of good sources, whose influence we aim to
maximize at the output, and with a set of bad sources, whose power we try to
minimize at the output. Interestingly, this leads to the standard Max-SINR
beamformer with a structured steering vector and covariance matrix. Define the
combined noise and interference covariance matrix as
\begin{equation}
  \mK_{nq} \bydef \mK_n + \sigma_z^2 \left( \sum_{k=0}^{K^\prime} \va(\vq_k)\right) \left(\sum_{k=0}^{K^\prime} \va(\vq_k) \right)^\H,
\end{equation}
where $\mK_n$ is the covariance matrix of the noise term, and $\sigma_z^2$ is
the power of the interferer at a particular frequency.

Then the solution to \eqref{eq:rake-maxsinr} is given as

\begin{equation}
  \label{eq:rake-maxsinr-weights}
  \vw_{\text{R-SINR}} = \frac{\mK_{nq}^{-1} \mA_s \vec{1}}{\vec{1}^\H \mA_s^\H \mK_{nq}^{-1} \mA_s \vec{1}}.
\end{equation}

It is interesting to note that when $\mK_{nq} = \sigma^2 \mI_M$ (\emph{e.g.} no
interferers and iid noise), the Rake-Max-SINR beamformer reduces to $\mA_s
\vec{1} / \norm{\mA_s \vec{1}}$, which is exactly the Rake-DS
beamformer. This is analogous to the non-raking DS beamformer
\eqref{eq:delay_and_sum_weights}.


\subsection{Max-UDR Raking} 
\label{sub:max_esnr_raking}

Finally, it is interesting to investigate what happens if we choose the
weights that optimize the perceptually motivated $\UDR$ \cite{Bradley:2003jf,
Lochner:1964do}. The $\UDR$ expresses the fact that adding early reflections
(up to 50 ms in the RIR) is as good as adding the energy to the direct sound
in terms of speech intelligibility. The useful signal is a
\emph{coherent} sum of the direct and early reflected speech energy, so that

\begin{equation}
  \label{eq:esnr}
  \UDR = \frac{\E \sum_{k=0}^K \abs{\vw^\H \va(\vs_k) x}^2}{\E \abs{\sum_{k=0}^{K'} \vw^\H \va(\vq_k)z + \vw^\H \vn}^2},
\end{equation}
In applications $K$ is rarely large enough to cover all the reflections
occurring within 50 ms, simply because it is too optimistic to assume we know
all the corresponding image sources. Therefore, \eqref{eq:esnr} typically
underestimates the $\UDR$. 

Alas, because \eqref{eq:esnr} is specified in the frequency domain, it is
challenging to control whether the reflections in the numerator arrive before
or after the direct sound. Nevertheless, it is interesting to analyze it as it
provides a basis for future work on time-domain raking formulations, and a
meaningful evaluation of the raking algorithms presented in this paper.

To compute the Rake-Max-UDR weights, we solve the following program

\begin{equation}
  \label{eq:rake-maxperception}
  \maximize_{\vw \in \C^M} \frac{\E \sum_{k=0}^K \abs{\vw^\H \va(\vs_k)x}^2}{\E \abs{\vw^\H \sum_{k=0}^{K'} \va(\vq_k)z + \vw^\H \vn}^2}.
\end{equation}
This amounts to maximizing a particular generalized Rayleigh quotient,
\begin{equation}
  \label{eq:gen_rayleigh_udr}
  \frac{\vw^\H \mA_s \mA_s^\H \vw}{\vw^\H \mK_{nq} \vw}.
\end{equation}
Assuming that $\mK_{nq}$ has a Cholesky decomposition as $\mK_{nq} = \mC^\H
\mC$ we can rewrite the quotient \eqref{eq:gen_rayleigh_udr} as
\begin{equation}
  \frac{\wt{\vw}^\H (\mC^{-1})^\H \mA_s \mA_s^\H \mC^{-1} \wt{\vw}}{\wt{\vw}^\H \wt{\vw}},
\end{equation}
where $\wt{\vw} \bydef \mC \vw$. The maximum of this expression is
\begin{equation}
  \lambda_\text{max}((\mC^{-1})^\H \mA_s \mA_s^\H \mC^{-1}),
\end{equation}
where $\lambda_\text{max}(\, \cdot \, )$ denotes the largest eigenvalue of the
matrix in the argument. The maximum is achieved by the corresponding
eigenvector $\wt{\vw}_\text{max}$. Then the optimal weights are given as
\begin{equation}
  \vw_\text{R-UDR} = \mC^{-1} \wt{\vw}_\text{max}.
\end{equation}

\subsection{SINR Gain from Raking} 
\label{sec:noise_suppression}

Intuitively, if we have multiple sources of the desired signal scattered in
space, and we account for it in the design, we should do at least as well as
when we ignore the image sources. Let us see how large the gain can be for the Rake-Max-SINR beamformer. We have that
\begin{equation}
  \SINR = (\mA_s \vec{1})^\H \mK_{nq}^{-1} (\mA_s \vec{1}).
\end{equation}
Intuitively, the larger the norm of $\mA_s \vec{1}$, the better the $\SINR$ (as
$\mK_{nq}$ is positive). To explicitly see if there is any gain in using the
acoustic rake receiver, we should compare the standard Max-SINR beamformer with the Rake-Max-SINR, \emph{e.g.}, we should evaluate
\begin{equation}
  \label{eq:gain_ratio}
  \frac{\big( \sum_k \va(\vs_k) \big)^\H \mK_{nq}^{-1} \big(\sum_k \va(\vs_k)
  \big)}{\va(\vs_0)^\H
  \mK_{nq}^{-1} \va(\vs_0)}.
\end{equation}
One possible interpretation of \eqref{eq:gain_ratio} is that we ask whether
the steering vectors $\va(\vs_k)$ sum coherently or they cancel out. 

\rev{To answer this, assume that $\vs_k$, $0 \leq k \leq K$ are the desired
sources (true and image), and let $\beta \bydef \sum_{k=1}^K
(\alpha_k/\alpha_0)^2$, where $\alpha_k$ is the strength of the source $\vs_k$
received by the array. Then 
\begin{equation}
\label{eq:snr_gain_raking}
\E \left(\left\|\sum_{k=0}^K \va(\vs_k)\right\|^2 \right) \approx (1+\beta)\,
\E(\norm{\va(\vs_0)}^2),
\end{equation}
that is, we can expect an increase in the output $\SINR$ approximately by a
factor of $(1+\beta)$ when using the Rake-Max-SINR beamformer. This statement
is made precise in Theorem \ref{thm:SNR-gain} in the Appendix. It holds when
$\mK_{nq}$ has eigenvalues of similar magnitude, which is typically not the
case in the presence of interferers. However, we show in Section
\ref{sec:numerical_experiments} that with interferers present, the gains actually
increase.}

  A couple of remarks are in order:
\begin{enumerate}
  \item This result is in expectation; \rev{it says that on average, the $\SINR$
  will increase by a factor of $(1+\beta)$. In the worst case, the steering
  vectors $\va(\vs_k)$ can even cancel out so that the $\SINR$ decreases.} But the
  numerical experiments suggest that this is very rare in practice, and we can
  on the other hand observe large gains.
  \item We see that summing the phasors in $a_m(\vs_k)$ behaves as a
  two-dimensional random walk. It is known that the root-mean-square distance of a 2D
  random walk from the origin after $n$ steps is $\sqrt{n}$
  \cite{mccrea1940random}.
  \item Due to the far-field assumption in Theorem \ref{thm:SNR-gain}, the
   attenuations $\alpha_k$ are assumed to be independent of the microphones;
   in reality they do depend on the source locations. However, they also
   depend on a number of additional factors, for example wall attenuations
   and radiation patterns of the sources. Therefore, for simplicity, we
   consider them to be independent. One can verify that this assumption does
   not change the described trend.
\end{enumerate}

It is reassuring to observe the behavior suggested by
\eqref{eq:snr_gain_raking} in practice. Fig.~\ref{fig:snr_gain} shows the
comparison of the prediction by Theorem~\ref{thm:SNR-gain} with the $\SNR$
gains observed in simulated rooms. In this case, we are comparing the pure
$\SNR$ gain for white noise, without interferers. To generate
Fig.~\ref{fig:snr_gain}, we randomized the location of the source inside the
rectangular room. For simplicity we fixed the signal power as received by the
microphones to the same value for all the image sources, so that the expected
gain is $K+1$ in the linear scale. The curves agree near-perfectly with the
prediction of Theorem \ref{thm:SNR-gain}.

\begin{figure}
\centering
\includegraphics[width=3.5in]{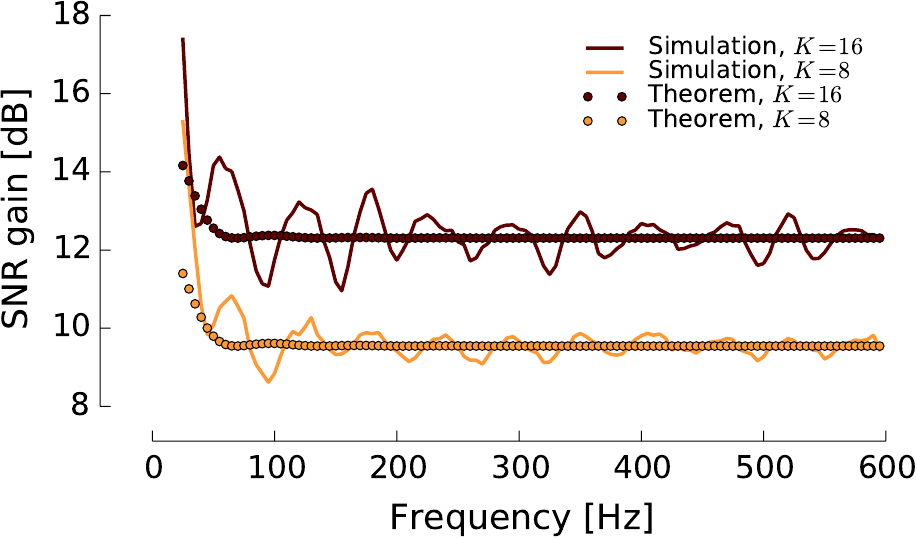}
\caption{Comparison of the simulated $\SNR$ gains and the theoretical
prediction from Theorem \ref{thm:SNR-gain} for $K = 8$, and $K = 16$. The
theoretical prediction of the gain is $10\log_{10}(8+1) \approx 9.54$ for $K=8$,
and $10\log_{10}(16+1) \approx 12.30$ for $K=16$.}
\label{fig:snr_gain}
\end{figure}

\section{Finding and tracking the echoes} 
\label{sec:finding_and_tracking_the_echoes}

\begin{figure}
\centering
\includegraphics[width=3.5in]{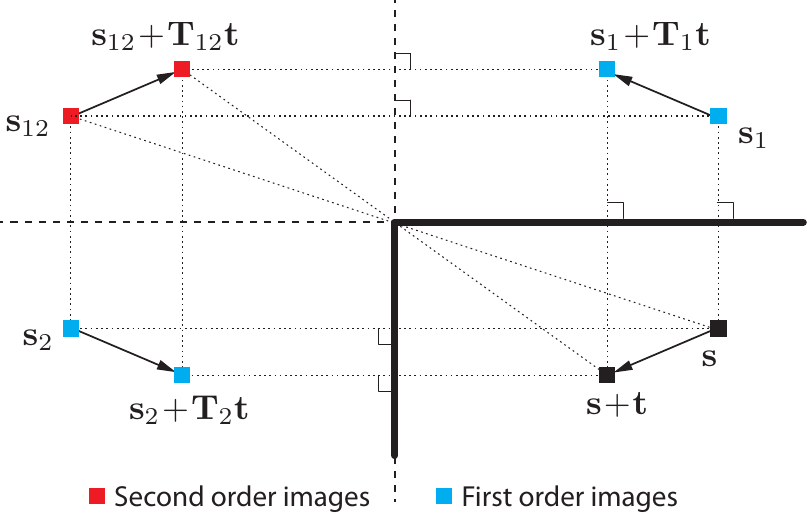}
\caption{Illustration of image source tracking in rectangular geometries.}
\label{fig:tracking}
\end{figure}

Thus far we assumed that the locations of the image sources are known. In this
section we briefly describe some methods to localize them when they are
\emph{a priori} unknown. We assume that we can localize the true source, or at
least one image source. Combined with the knowledge of the room geometry, this
suffices to find the locations of other image sources \cite{Ocal:2014wx}.

\subsection{Known room geometry} 
\label{sub:known_room_geometry}

In many cases, for example for fixed deployments, the room geometry is known.
This knowledge could be obtained at the time of the deployment, or from
blueprints. In most indoor environments, we  encounter a large number of
planar reflectors. These reflectors  correspond to image sources. With
reference to Fig.
\ref{fig:tracking}, we can easily compute the image source locations
\cite{Allen:1979ua} (we note that the image source model is not limited to
right angle geometries \cite{Borish:1984uu}).

Suppose that the real source is located at $\vs$. Then the image source with
respect to wall $i$ is computed as,
\begin{equation}
\label{eq:image-sources}
\vec{im}_i(\vs) = \vs + 2 \inprod{\vp_i - \vs, \vn_i} \vn_i,
\end{equation}
where $i$ indexes the wall, $\vn_i$ is the outward normal associated with the
$i$th wall, and $\vp_i$ is any point belonging to the $i$th wall. Analogously,
we compute image sources corresponding to higher order reflections,
\begin{equation}
  \vec{im}_j(\vec{im}_i(\vs)) = \vec{im}_i(\vs) + 2 \inprod{\vp_j - \vec{im}_i(\vs),
    \vn_j} \vn_j.
\end{equation}
The above expressions are valid regardless of the dimensionality, concretely
in 2D and 3D.


\subsection{Acoustic geometry estimation} 
\label{sub:acoustic_geometry_estimation}

When the room geometry is not known, it is possible to estimate it using the
same array that we use for beamforming. Recently a number of different methods
appeared in the literature that propose to use sound to estimate the shape
of a room. For example, in \cite{Ribeiro:2012ge} the authors use a dictionary
of wall impulse responses recorded with a particular array. In
\cite{Antonacci:2012hc} the authors use tools from projective geometry
together with the Hough transform to estimate the room geometry. In
\cite{Dokmanic:2013dz} the authors derive an \emph{echo sorting} mechanism
that finds the image sources, from which the room geometry is then derived.


\subsection{Without Estimating the Room Geometry} 
\label{sub:not_estimating_the_room_geometry}

To design an ARR, we do not really need to know how the room looks like; we
only need to know where the major echoes are coming from. One possible
approach is to locate the image sources in the initial calibration phase, and
then track their movement by tracking the true source.

We propose a tracking rule that leverages the knowledge of the displacement of
the true source. Again with reference to Fig. \ref{fig:tracking}, we can state
the following simple proposition.

\begin{prop}
  \label{prop:is_tracking}
  Suppose that the room has only right angles so that the walls are parallel
  with the coordinate axes. Let the source move from $\vs$ to $\vs +
  \vt$. Then any image source $\vs_k$, moves to a point given by
  \begin{equation}
    \vs_k + \mT \vt,
  \end{equation}
  where $\mT = \diag(\pm 1, \mp 1)$ for odd generations, and $\mT = \pm \mI_2$
  for even generations.
\end{prop}

\begin{proof}
  The proof follows directly from the figure. The displacement of the image
  source is the same as the displacement of the true source, passed through a
  series of reflections. Reflection matrices are diagonal matrices with $\pm
  1$ on the diagonal, and determinant equal to $-1$, hence the result.
\end{proof}

The usefulness of this proposition is that it gives us a tool to track the
image sources even when we do not know the room geometry (as long as it has
right angles). A possible use scenario is to start with a calibration
procedure with a controlled source, and perform the echo sorting to find
multiple image sources. Then if possible, we assign to each image source a
generation (this is in fact a by-product of echo sorting), or we try different
hypotheses using Proposition \ref{prop:is_tracking}, and choose the one that
maximizes the output $\SINR$.



\section{Numerical Experiments} 
\label{sec:numerical_experiments}

\begin{figure}
  \centering
  \includegraphics[width=3.5in]{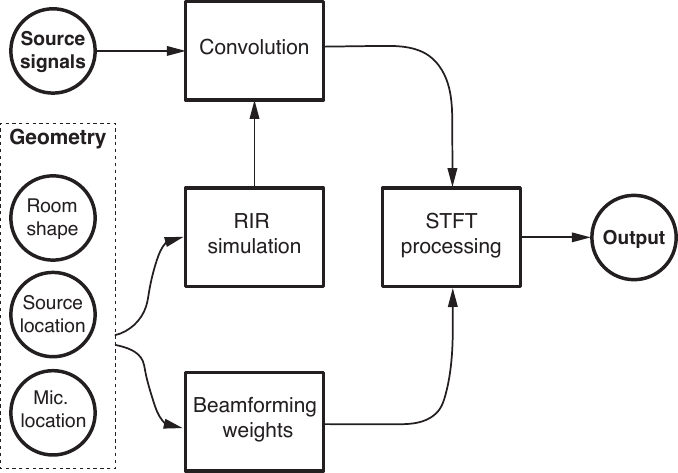}
  \caption{Block diagram of the simulation setup used for numerical experiments.}
  \label{fig:simulation_setup}
\end{figure}

In this section, we validate the described theoretical results through
numerical experiments. First, we analyze the beampatterns produced by the ARR;
second, we evaluate the $\SINR$ for various beamformers as a function of the
number of image sources used in weight computation; and third, we evaluate the
PESQ metric \cite{Rix:2001bv}. Finally, we show spectrograms that reveal
visually the improved interferer and noise suppression achieved by the ARR.

\subsection{Simulation Setup} 
\label{sec:simulation_setup}

We use a simple room acoustic framework written in Python, that relies on
Numpy and Scipy for matrix computations \cite{Oliphant:2007dm}. We limit
ourselves in this paper to 2D geometry and rectangular rooms. In all
experiments, the sampling frequency $F_s$ was set to 8 kHz. An overview of the
simulation setup is shown in Fig.~\ref{fig:simulation_setup}.

\begin{figure*}
  \centering
  \includegraphics[width=7in]{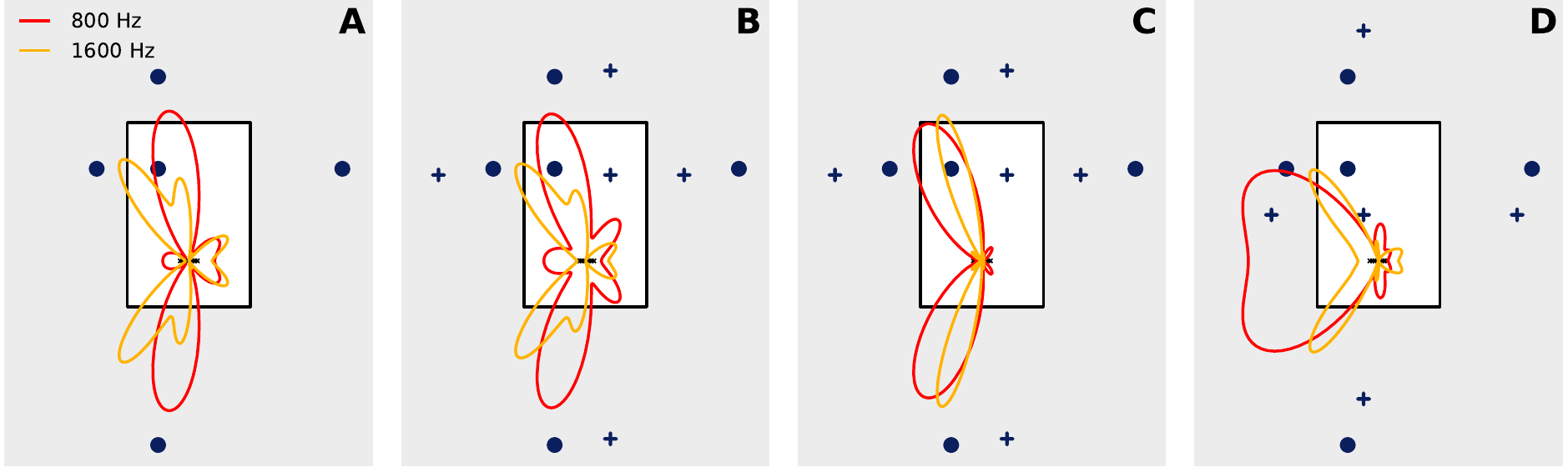}
  \caption{Beam patterns in different scenarios. The rectangular room is 4 by 6
    metres and contains a source of interest ($\CIRCLE$) and an interferer
    ({\bf +}) ((B), (C), (D) only). The first order
    image sources are also displayed. The weight computation of the beamformer
    includes the direct source and the first order image sources of both
    desired source and interferer (when applicable). (A) Rake-Max-SINR, no
    interferer, (B) Rake-Max-SINR, one interferer, (C) Rake-Max-UDR, one
    interferer, (D) Rake-Max-SINR, interferer is in direct path.} 
  \label{fig:beam_scenarios} 
\end{figure*}

Starting from the room geometry and the positions of the sources and
microphones, we first compute the locations of all images sources up to ten
generations.  The reflectivity of the walls is fixed to 0.9. The RIR between
the source $\vs_0$ and the microphone $\vr_m$ is convolved with an ideal
low-pass filter in the continuous domain and then sampled at the sampling
frequency $F_s$,
\begin{equation}
  \label{eq:rir_lowpass}
  \wt{a}_m(\vs_0)[n] = 
  \sum_{k=0}^K \frac{\alpha_k}{4\pi\|\vr_m-\vs_k\|}\,\sinc\left(n - F_s\frac{\|\vr_m-\vs_k\|}{c}\right),
\end{equation}
where $K$ is the number of image sources considered. We choose the limits of
$n$ such that the cardinal sine in \eqref{eq:rir_lowpass} decays sufficiently
to avoid artifacts. The discrete signals from all sound sources are then
convolved with their respective RIRs, and added together to obtain the $m$th
microphone's signal.


The beamforming weights are computed in the frequency domain. We use the
discrete-time STFT processing with a frame size of $L=4096$ samples, 50\%
overlap and zero padding on both sides of the signal by $L/2$. A real fast
Fourier transform of size $2L$ and a Hann window are used in the analysis. By
exploiting the conjugate symmetry of the real FFT we only need to compute
$L+1$ beamforming weights, one for every positive frequency bin. The length
$L$ is dictated by the length of the beamforming filters in the time-domain
and was set empirically to avoid any aliasing in the filter responses. The
output signal was synthesized using the conventional overlap-add method
\cite{Shynk:1992vh}.


\subsection{Results} 
\label{sec:results}

\subsubsection{Beampatterns}

We first inspect the beampatterns produced by the Rake-Max-SINR and
Rake-Max-UDR beamformers for different source-interferer placements.  We
consider a $4 \text{m} \times 6 \text{m}$ rectangular room with a source of
interest at $(1 \text{m},4.5 \text{m})$ and a linear microphone array centered
at $(2\text{m},1.5\text{m})$, parallel to the $x$-axis. Spacing between the
microphones was set to 8cm. In Fig.~\ref{fig:beam_scenarios}, we show the
beampatterns for four different configurations of the source and the
interferer. We consider a scenario without an interferer, one with an
interferer placed favorably at $(2.8\text{m},4.3\text{m})$, and finally one
where the interferer is placed half-way between the desired source and the
array, at $(1.5\text{m},3\text{m})$.

The last scenario is the least favorable. Interestingly, we can observe that
the Rake-Max-SINR beampattern adjusts by completely ignoring the direct path,
and steering the beam towards the echoes of the source of interest. This is
validating the intuition that we can ``hear behind an interferer by listening
for the echoes''. Note that such a pattern cannot be achieved by a beamformer
that only takes into account the direct path. We further note that, while the
beampatterns only show the magnitude of the beamformer's response, the phase
plays an important role with multiple sources present.

\subsubsection{SINR Gains from Raking}

In the experiments in this subsection, we set the power of the desired source
and of the interferer to be equal, $\sigma_x^2 = \sigma_z^2 = 1$. The noise
covariance matrix is set to $10^{-3} \cdot \mI_M$. We use a circular array of
$M=12$ microphones with a diameter of 30 cm, and randomize the position of the
desired source and the interferer inside the room. The resulting curves show
median performance out of 20000 runs.


Fig. \ref{fig:SINR_vs_K} shows output $\SINR$ for different beamformers. The
one-forcing beamformer is left out because it performs poorly in terms of
$\SINR$, as predicted earlier. Clearly, the Rake-Max-SINR beamformer
outperforms all others. The output $\SINR$ for beamformers using only the
direct path (Max-SINR and DS) remains approximately constant. The $\UDR$ is
plotted against the number of image sources for various beamformers in
Fig.~\ref{fig:UDR_vs_K}. The Rake-Max-UDR beamformer performs well in terms of
the two measures, its output is perceptually unpleasing due to audible
pre-echoes; in informal listening tests, the Rake-Max-SINR beamformer did not
produce such artifacts. It is interesting to note that the Rake-Max-SINR also
performs well in terms of the $\UDR$. Similar $\SINR$ gains to those in
Fig.~\ref{fig:SINR_vs_K} are observed in Fig.~\ref{fig:SINR_vs_freq} over a
range of frequencies. It is therefore justified to extrapolate the results at
one frequency in Fig.~\ref{fig:SINR_vs_K} to the wideband $\SINR$.

\rev{
\subsubsection{Evaluation of Speech Quality}

We complement the informal listening tests and the evaluation of $\SINR$ and
$\UDR$ with extensive simulations to asses the improvement in speech quality
achieved by acoustic raking. We simulate a room with two sources---a desired
source and an interferer---and compare the outputs of the Rake-DS,
Rake-Max-SINR, and Rake-Max-UDR as a function of the number of image sources
used to design the beamformers.

The same number of image sources is used for the target and interferer
($K=K^\prime$). The performance metric used is the PESQ
\cite{Rix:2001bv}. In particular, we use the reference implementation
described by the the ITU P.862 Amendment 2
\cite{P.862a2}. PESQ compares the reference signal with the degraded signal and
predicts the perceptual quality of the latter as it would be measured by the
mean opinion score (MOS) value, on a scale from $1$ to $4.5$.

\begin{figure}
\centering
\includegraphics[width=3.5in]{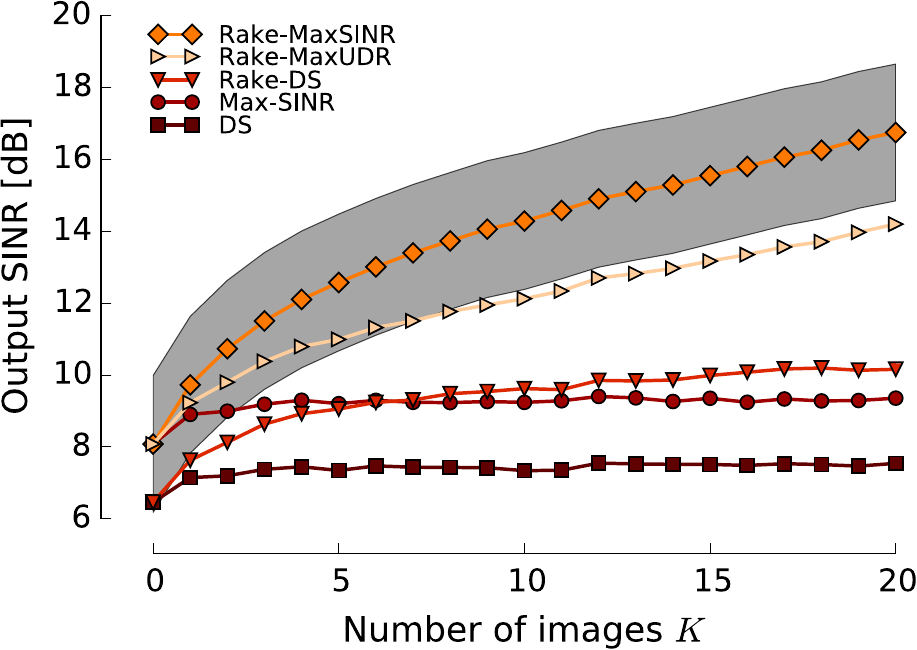}
\caption{Median output $\SINR$ plotted against the number of image sources
used in the design for different beamformers, at a frequency $f = 1$~kHz. The
shaded area contains the Rake-Max-SINR output $\SINR$s for 50\% of the 20000
Monte Carlo runs.}
\label{fig:SINR_vs_K}
\end{figure}

\begin{figure}
\centering
\includegraphics[width=3.5in]{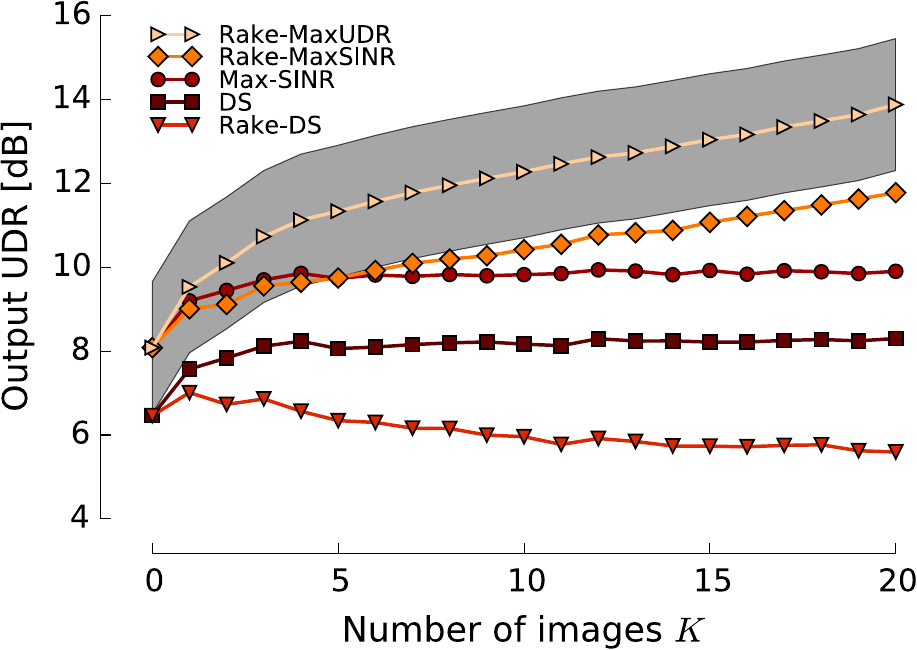}
\caption{Median output $\UDR$ plotted against the number of image sources used in
the design for different beamformers, at a frequency $f = 1$~kHz. The shaded
area contains the Rake-Max-UDR output $\UDR$s for 50\% of the 20000 Monte
Carlo runs.}
\label{fig:UDR_vs_K}
\end{figure}

We consider the same room and microphone array setting as before (see
Fig.~\ref{fig:beam_scenarios}A). The desired and the interfering sources are
placed uniformly at random in a rectangular area with lower left corner at
$(1\text{m},2.5\text{m})$ and upper right corner at $(3\text{m},5\text{m})$.
To limit the experimental variation, the speech samples attributed to the
sources are fixed throughout the simulation. The two sources start reproducing
speech at the same time and approximately overlap for the total duration of
the speech samples. The signals are normalized to have the same power at the
source. We added white Gaussian noise to the microphone signals, with power chosen
so that the $\SNR$ of the direct sound for the desired source is 20 dB at the
center of the microphone array. All signals are high-pass filtered with a
cut-off frequency of 300 Hz. The reference for all PESQ results is the direct
path of the target signal as measured at the center of the array
$(2\text{m},1.5\text{m})$.

\begin{figure}[h!]
\centering
\includegraphics[width=3.5in]{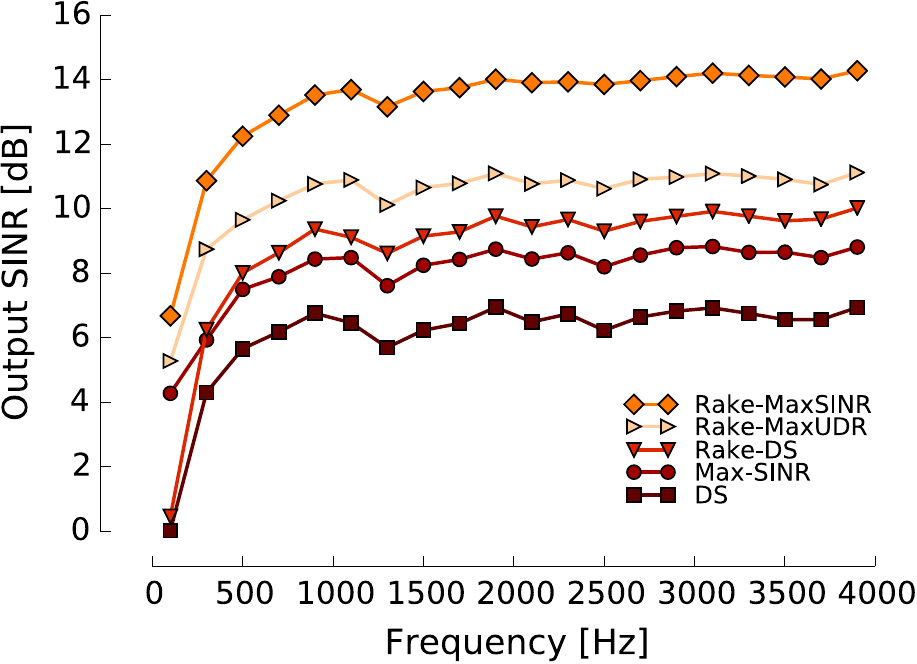}
\caption{Output $\SINR$ as a function of frequency for different beamformers and
$K=K'=10$. The curves show the average of 20000 runs, with averaging performed
in the dB domain.}
\label{fig:SINR_vs_freq}
\end{figure}

The median PESQ measure of 10000 Monte Carlo runs, given in raw MOS, is shown
in Fig.~\ref{fig:perceptual_quality}. The median PESQ of the degraded signal
measured at the center of the array before processing was found to be 1.6 raw
MOS. When only the direct sound is used (\emph{i.e.}, $K=0$), all three
beamformers yield the same improvement of about 0.2 raw MOS. We observe that
Rake-DS is marginally better than the other beamformers. Using any number of
echoes in addition to the direct sound results in larger MOS for all
beamformers. When more than one image source is used, the Rake-Max-SINR
beamformer always yields the largest MOS, with up to 0.5 MOS gain when using
10 images sources.

It is worth mentioning that in the beamformer design, we do not assume that we
know the spectrum of the source or the interferer---we design as if it was
flat. Thus the interferer acts as a strong source of colored, spatially
correlated, non-stationary noise, spectrally mismatched with the designed
beamformer. There is another source of model mismatch: while the RIRs were
computed using hundreds of image sources, we use only up to ten to design the
beamformers.}

\subsubsection{Spectrograms and Sounds Samples}

Finally, we present the spectrograms for a scenario where we want to focus on
a singer in the presence of interfering speech. We consider the same room,
source, interferer, and microphone array geometry as in
Fig.~\ref{fig:beam_scenarios}B.

The source signal is a snippet by a female opera singer
(Fig.~\ref{fig:spectrograms}A), with strongly pronounced harmonics; the
interfering signal is a male speech extract. The two signals are normalized to
have unit maximum amplitude. We add white Gaussian noise to the microphone
signals with power such that the SNR of the direct sound of the desired source
is 20 dB at the center of the microphone array. All signals are high-pass
filtered with a cut-off frequency of 300 Hz. The Rake-Max-SINR beamformer
weights are computed using the direct source and three generations of image
sources for both the desired sound source (singing) and the interferer
(speech).

The output of the conventional Max-SINR beamformer
(Fig.~\ref{fig:spectrograms}C) is compared to that of the Rake-Max-SINR
(Fig.~\ref{fig:spectrograms}D). We can observe from the spectrogram that the
Rake-Max-SINR reduces very effectively the power of the interfering signal at
all frequencies, but particularly in the mid to high range. This is true even
when the interferer overlaps significantly with the desired signal. Informal
listening tests confirm that the Rake-Max-SINR maintains high quality of the
desired signal while strongly reducing the interference. The Rake-Max-UDR
beamformer provides good interference suppression, but it produces audible
pre-echoes that render it unsuitable for speech processing applications. The
sound clips can be found online along the code.


\section{Conclusion} 
\label{sec:conclusion}

\begin{figure}[t!]
\centering
\includegraphics[width=3.5in]{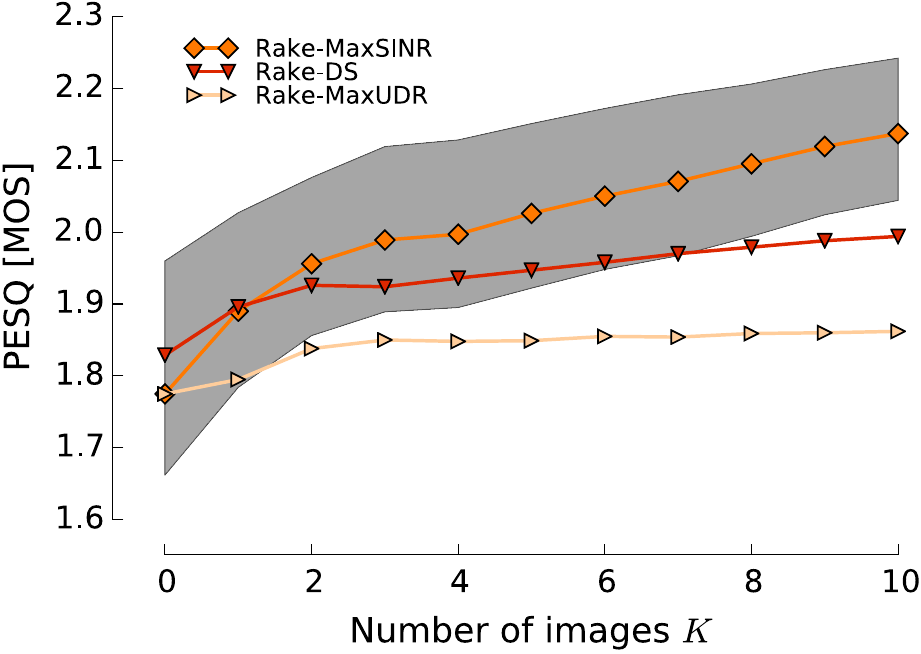}
\caption{Perceptual quality in MOS, evaluated using PESQ, as a function of the
  number of image sources used $K$. The lower limit of the ordinates is set to
  the median MOS of the degraded signal before processing, as measured at
  center of the array. The shaded area contains the Rake-Max-SINR output for
  50\% of the 10000 Monte Carlo runs.}
\label{fig:perceptual_quality}
\end{figure}

We investigated the concept of acoustic rake receivers---beamformers that use
echoes. Unlike earlier related work, we presented optimal formulations that
outperform the delay-and-sum style approaches by a large margin. This is
especially true in the presence of interferers, hence the title ``Raking the
Cocktail Party''. We demonstrate theoretically that the ARRs improve the $\SINR$,
and the numerical simulations agree well with these predictions.

Beyond theoretical and numerical evaluations of the performance measures, we
demonstrated in informal listening tests the improved interference suppression
by the ARR.  Moreover, extensive simulation determined that the ARR improves
the subjective quality, as predicted by PESQ, proportionnaly to the number of
image sources used. A particularly illustrative example is when the interferer
is occluding the desired source---the optimal ARR takes care of this simply by
listening to the echoes.

\begin{figure*}
  \centering
  \includegraphics[width=7in]{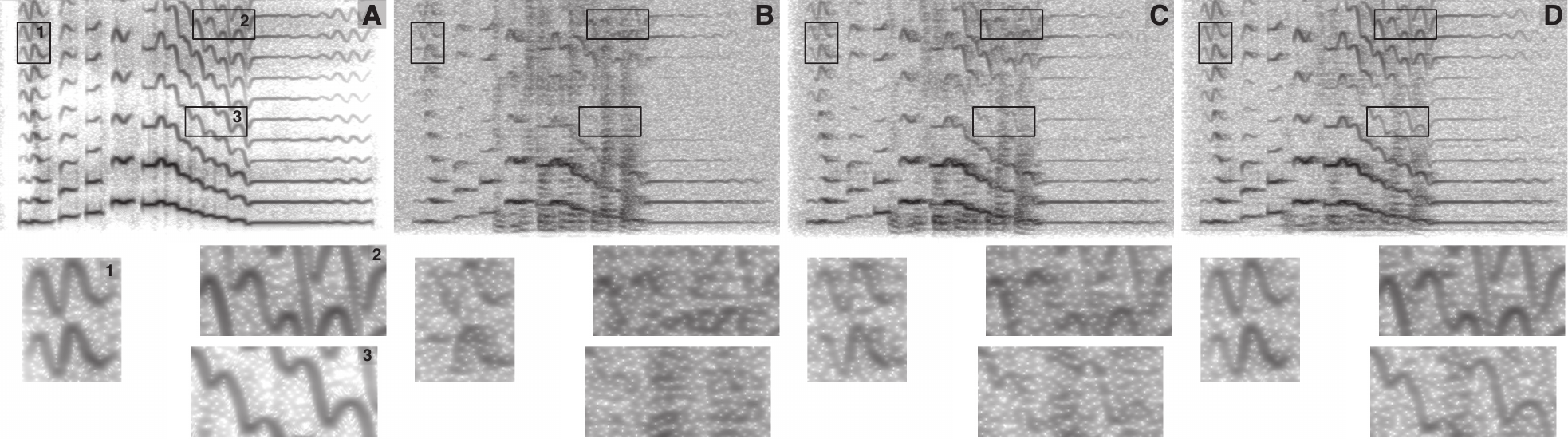}
  \caption{Comparison of the conventional Max-SINR and Rake-Max-SINR beamformer
    on a real speech sample. Spectrograms of (A) clean signal of
    interest, (B) signal corrupted by an interferer and additive white
    Gaussian noise at the microphone input, outputs of (C) conventional
    Max-SINR and (D) Rake-Max-SINR beamformers. Time naturally 
    goes from left to right, and frequency increases from zero at the 
    bottom up to $F_s/2$. To highlight the improvement of Rake-Max-SINR over
    Max-SINR, we blow-up three parts of the spectrograms in the lower part of
    the figure. The boxes and the corresponding part of the original spectrogram
    are numbered in (A). The numbering is the same but omitted in the rest of
    the figure for clarity.}
  \label{fig:spectrograms}
\end{figure*}

Perhaps the most important aspect of ongoing work is the design of robust
formulations of ARRs. This may involve various heuristics, as well as
combinatorial optimization due to the discrete nature of image sources.
\rev{We expect that the raking beamformers described in this paper inherit the
robustness properties of their classical counterparts. For example, the Rake-%
DS beamformer is likely to be more robust to array calibration errors than the
Rake-Max-SINR beamformer. Furthermore, we expect that taking the image source
perspective makes various ARRs more robust to errors in source locations than
the schemes that assume the knowledge of the RIR.}

\rev{Another line of ongoing work investigates the time-domain formulations of the
ARRs, with some initial results already available \cite{Scheibler:2015xx}.
Time-domain formulations offer better control over whether the echoes appear
before or after the direct sound. This provides a more natural framework for
optimizing perceptually motivated performance measures, such as $\UDR$.
}






\section*{Appendix: Theorem \ref{thm:SNR-gain}}

We note that the theorem is is stated for a linear array, but the described
behavior is universal.

\begin{thm}
  \label{thm:SNR-gain} Assume that there are $K+1$ sources located at $\vs_k =
  r_k [\cos
  \theta_k ~ \sin \theta_k]^\T$ where $\theta_k \sim \mathcal{U}(0, 2\pi)$ and
  $r_k \sim \mathcal{U}(a, b)$ are all independent, for some $0<a<b$ such that
  the far-field assumption holds. Let $\mA_s$ collect the corresponding
  steering vectors for a uniform linear microphone array. Then $\E \norm{\mA_s
  \vec{1}}^2 \geq (1+\beta)\, \E \norm{\va(\vs_0)}^2$, where $\beta =
  \sum_{k=1}^K (\alpha_k/\alpha_0)^2$, and $\alpha_k$ are attenuations of the
  steering vectors, assumed independent from the source locations. In fact,
  $\E(\norm{\mA_s \vec{1}}^2) = (1+\beta)\, \E(\norm{\va(\vs_0)}^2) + O(1 /
  \Omega^3)$.

\end{thm}

\begin{proof}
  Thanks to the far-field assumption, we can decompose the steering vector into a
  factor due to the array, and a phase factor due to different distances of
  different image sources. We have that
  \begin{equation}
    a_m = (\mA_s \vec{1})_m = \sum_{k=0}^K \alpha_k \eni{\kappa m d \sin \theta_k} \eni{\Omega \delta_k/c},
  \end{equation}
  where $d$ is the microphone spacing and $\kappa \bydef \Omega/c$. Without
  loss of generality we assume that $\delta_k \sim \mathcal{U}(a, b)$. We can
  further write
  \begin{equation}
  \label{eq:sumnorm_proof}
  \begin{aligned}
    \E \abs{a_m}^2
    &= \E~\bigg[\bigg( \sum_{k=0}^K \alpha_k \eni{\kappa m d \sin \theta_k} \eni{\kappa \delta_k} \bigg) \\
    &\phantom{= \E} \times \bigg( \sum_{\ell=0}^K \alpha_\ell \epi{\kappa m d \sin \theta_\ell} \epi{\kappa \delta_\ell} \bigg) \bigg] \\
    &\hspace{-12mm} = \sum_{k=0}^K \alpha_k^2 \!+\!\!\! \sum_{k \neq \ell=0}^K \!\! \alpha_k \alpha_\ell \E \left[ \epi{\kappa m d \left( \sin \theta_\ell - \sin \theta_k \right)} \epi{\kappa (\delta_\ell - \delta_k)}  \right].
  \end{aligned}
  \end{equation}
  Invoking the independence for $k \neq \ell$, we compute the above
  expectation as
  \begin{equation}
  \begin{aligned}
     &\E \left[  \epi{\kappa m d \left( \sin \theta_\ell - \sin \theta_k \right)} \epi{\kappa (\delta_\ell - \delta_k)} \right] \\
     & \hspace{3cm} = \frac{2 J_0^2(md\kappa)\big[1-\cos (\Delta \kappa)\big]}{(\Delta \kappa)^2},
  \end{aligned}
  \end{equation}
  where $J_0$ denotes the Bessel function of the first kind and zeroth order and
  $\Delta \bydef b-a$.

  Plugging this back into \eqref{eq:sumnorm_proof}, we obtain
  \begin{equation}
    \E \abs{a_m}^2 = \sum_{k=0}^K \alpha_k^2
    \left( 1 + 
      C
      \frac{2 J_0^2(md\kappa) \big[1-\cos (\Delta \kappa)\big]}{(\Delta \kappa)^2}\right),
  \end{equation}
  where $C = \sum_{k\neq \ell} \alpha_k\alpha_\ell / \sum_k \alpha_k^2$.
  
  Because $\abs{J_0(z)} \leq \sqrt{2/(\pi z)} + O(\abs{z}^{-1})$
  (\cite{Abramowitz:1972vw}, Eq. 9.2.1), we see that the
  expression in brackets is $1 + O(\Omega^{-3})$. Rewriting
  \begin{equation}
    \sum_{k=0}^K \alpha_k^2 = \frac{1}{M}\E \norm{\va(\vs_0)}^2
    \left(1 + \sum_{k=1}^K (\alpha_k/\alpha_0)^2\right)
  \end{equation}
  concludes the proof.
\end{proof}

\section*{Acknowledgments}

We thank the anonymous reviewers for a number of constructive suggestions that
have improved the manuscript.

\balance

\bibliographystyle{IEEEbib}
\bibliography{acurake-jstsp}

\end{document}